\documentclass[a4paper,numberwithinsect]{article}
\usepackage[utf8]{inputenc}

\usepackage{amsthm,amssymb,amsfonts,amsmath}
\usepackage{amssymb,amsfonts,amsmath}
\newtheorem{theorem}{Theorem}
\newtheorem{lemma}[theorem]{Lemma}
\newtheorem{corollary}[theorem]{Corollary}
\newtheorem{proposition}[theorem]{Proposition}
\newtheorem{observation}[theorem]{Observation}
\theoremstyle{definition}
\newtheorem{definition}[theorem]{Definition}

\usepackage[hidelinks]{hyperref}
\usepackage[blocks]{authblk}

\usepackage{graphicx}
\graphicspath{{figures/}}
\usepackage{wrapfig}
\usepackage{caption,subcaption}

\makeatletter
\g@addto@macro\bfseries{\boldmath}
\makeatother

\advance\hoffset-7mm
\usepackage{thm-restate}

\usepackage{url}

\newcommand{\R}{\mathbb{R}}

\newcommand{\s}{\mathcal{S}}
\newcommand{\q}{\hat{q}}
\newcommand{\A}{\mathcal{A}}

\makeatletter
\g@addto@macro\bfseries{\boldmath}
\makeatother

\usepackage{authblk}

\date{March 2020}
\title{Tukey Depth Histograms}

\begin{document}
\author[1]{Daniel Bertschinger}
\affil[1]{Department of Computer Science, ETH Z\"urich\\ \texttt{daniel.bertschinger@inf.ethz.ch}}

\author[2]{Jonas Passweg}
\affil[2]{Department of Computer Science, ETH Z\"urich\\ \texttt{jonas.passweg@student.ethz.ch}}

\author[3]{Patrick Schnider}
\affil[3]{Department of Mathematical Sciences, University of Copenhagen\\ \texttt{ps@math.ku.dk}}


\maketitle

\begin{abstract}
The Tukey depth of a flat with respect to a point set is a concept that appears in many areas of discrete and computational geometry.
In particular, the study of centerpoints, center transversals, Ham Sandwich cuts, or $k$-edges can all be phrased in terms of depths of certain flats with respect to one or more point sets.
In this work, we introduce the \emph{Tukey depth histogram} of $k$-flats in $\mathbb{R}^d$ with respect to a point set $P$, which is a vector $D^{k,d}(P)$, whose $i$'th entry $D^{k,d}_i(P)$ denotes the number of $k$-flats spanned by $k+1$ points of $P$ that have Tukey depth $i$ with respect to $P$.
As our main result, we give a complete characterization of the depth histograms of points, that is, for any dimension $d$ we give a description of all possible histograms $D^{0,d}(P)$. This then allows us to compute the exact number of possible such histograms.
\end{abstract}

\section{Introduction}

Many fundamental problems on point sets, such as the number of extreme points, the number of halving lines, or the crossing number do not depend on the actual location and distances of the points, but rather on some underlying combinatorial structure of the point set.
There is a vast body of work of combinatorial representations of point sets, at the beginning of which are the seminal series of papers by Goodman and Pollack \cite{goodman2, goodman1, semispaces}, where many important objects such as \emph{allowable sequences} and \emph{order types} are introduced.
In particular order types have proven to be a very powerful representation of point sets.
For many problems however, less information than what is encoded in order types is sufficient.
One example for such a problem is the determination of the \emph{depth} of a query point with respect to a planar point set.

Depth measures are a tool to capture how deep a query point lies within a given point set.
There is a number of depth measures that have been introduced, most famously \emph{Tukey depth} \cite{Tuk75} (also called \emph{halfspace depth}), \emph{Simplicial depth} \cite{LiuSimplicial} or \emph{Convex hull peeling depth} (see \cite{Aloupis2003GeometricMO, Hugg} or Chapter 58 in \cite{Handbook} for an overview of depth measures).
In this paper, we are mainly concerned with Tukey depth.
The Tukey depth of a query point $q$ with respect to a point set $P$ is the minimum number of points of $P$ that lie in a closed halfspace containing $q$.
For Tukey depth (and Simplicial depth), the depth of a query point in the plane can be computed knowing only the \emph{line rotational order}.

The line rotational order of the points of a point set $P$ around a query point $q$ is the order in which a directed line rotating around $q$ passes over the points of $P$, where we distinguish whether a point of $P$ is passed in front of, or behind $q$.
In \cite{semispaces}, this is called the \emph{local sequence of ordered switches}, and it is shown that knowing this local information for every point in a point set, one can uniquely determine the order type.\footnote{This is in contrast to the \emph{ray rotational order}, also called \emph{rotation system}, which is defined as the order in which a ray rotating around $q$ passes over the points of $P$.
Knowing only the ray rotational order around every point, one can generally not reconstruct the order type, in fact there can be up to $n-1$ different order types of $n$ points that give rise to the same ray rotational orders \cite{aichholzer}.}

In fact, the Tukey depth of a query point $q$ in the plane can be computed using even less information than the line rotational order around $q$: it suffices to know for each $k$, how many directed lines through $q$ and a point of $P$ have exactly $k$ points to their left.
This defines the \emph{$\ell$-vector} of $q$.
The Tukey depth of $q$ is now just the smallest $k$ for which the corresponding entry in the $\ell$-vector is non-zero.
It turns out, that many other depth measures can also be computed knowing only the $\ell$-vector of $q$ \cite{durocher}.
Another quantity that can be computed from this information only is the number of crossing-free perfect matchings on $P\cup\{q\}$, if $P$ is in convex position and $q$ is in the convex hull of $P$ \cite{Ruiz-Vargas2017}.
In \cite{Ruiz-Vargas2017}, a characterization of all possible $\ell$-vectors is given, phrased in terms of \emph{frequency vectors}, which is an equivalent object.
Knowing the $\ell$-vector of every point in a point set $P$ thus still gives us a lot of information about this point set.
For example, as this allows us to compute the simplicial depth $\text{sd}(P,q)$ of each point $q$ in $P$, that is, the number of triangles spanned by $P\setminus \{q\}$ that have $q$ in their interior, this also allows us to compute the crossing number of $P$, which is just
\[ \text{cr}(P)=\binom{n}{4}-\sum_{q\in P}\text{sd}(P,q). \]

Interesting objects emerge after forgetting yet another piece of information: instead of knowing the $\ell$-vector of each point, assume we only know the sum of all $\ell$-vectors.
This corresponds to knowing for each $j$ the number of \emph{$j$-edges} that is, knowing the histogram of $j$-edges.
The number of $j$-edges that a point set admits is a fundamental question in discrete geometry and has a rich history, see e.g. \cite{wagner_survey}, Chapter 4 in \cite{felsner_book} or Chapter 11 in \cite{matousek_book} and the references therein.

In this work, we mainly investigate a similar concept: depth histograms of points.
This corresponds to knowing for each $j$ how many points of Tukey depth $j$ are in the point set.
We give a complete characterization of possible such histograms for point sets in general position.
In particular, we will show the following:
\begin{theorem}\label{th:defhisto}
A vector $D^{0,d}$ is a depth histogram of a point set in general position in $\R^d$ if and only if for all nonzero entries $D^{0,d}_i$ with $i \geq 2$ we have \[\sum_{j=1}^{i-1} D^{0,d}_j \geq 2i + d - 3.\] 
\end{theorem}

In fact, both depth histograms of points as well as $j$-edges can be viewed as instances of a more general definition, that of histograms of $j$-flats, which we will introduce in the following.
For some of our results we give versions in this general setting.
We hope that this work can serve as a small step in the systematic study of these histograms.

In order to define histograms of $j$-flats, we first define the Tukey depth of a flat:

\begin{definition}
Let $Q$ be a set of $k+1$ points in $\mathbb{R}^d$, $k<d$, which span a unique $k$-flat $F$.
The \emph{affine Tukey depth} of $Q$ with respect to a point set $P$, denoted by $\emph{atd}_P(Q)$, is the minimum number of points of $P$ in any closed halfspace containing $F$.
The \emph{convex Tukey depth} of $Q$ with respect to $P$, denoted by $\emph{ctd}_P(Q)$, is the minimum number of points of $P$ in any closed halfspace containing $\text{conv}(Q)$.
\end{definition}

Note that for $k=0$ both definitions coincide with the standard definition of Tukey depth, and we just write $\emph{td}_P(q)$ in this case.
Further note that if $P\cup Q$ is in convex position, then $\emph{atd}_P(Q)=\emph{ctd}_P(Q)$.
Several results in discrete geometry can be phrased in terms of this generalized Tukey depth.
For example, the center transversal theorem \cite{CT_Dolnikov, CT_Zivaljevic} states that for any $j+1$ point sets $P_0,\ldots, P_j$ in $\mathbb{R}^d$, there exists a $j$-flat (not necessarily spanned by points of the point sets) that has affine Tukey depth $\frac{|P_i|}{d+1-j}$ with respect to $P_i$, for each $i\in\{0,\ldots,j\}$.
For $j=0$ and $j=d-1$, we retrieve the centerpoint theorem \cite{CP} and Ham-Sandwich theorem \cite{HS} as boundary cases.

\begin{definition}
Let $P$ be a set of points in $\mathbb{R}^d$.
The \emph{affine Tukey depth histogram of $j$-flats}, denoted by $D^{j,d}(P)$, is a vector whose entries $D^{j,d}_k(P)$ are the number of subsets $Q\subset P$ of size $j+1$ whose affine Tukey depth is $i$.
Similarly, replacing affine Tukey depth with convex Tukey depth, we define the \emph{convex Tukey depth histogram of $j$-flats}, denoted by $cD^{j,d}(P)$.
\end{definition}

In the following, we will also call affine Tukey depth histograms just \emph{depth histograms}, that is, unless we specify the \emph{convex}, we always mean an affine Tukey depth histogram.
Note however that for $j=0$ or if $P$ is in convex position, the two histograms coincide.

Many problems in discrete geometry can also be phrased in terms of depth histograms.
For example, the number of extreme points of a point set $P$ just corresponds to the entry $D^{0,d}_1(P)$ (note that each point of $P$ has Tukey depth at least 1).
Further, the number of $j$-edges or, more generally, $j$-facets corresponds to the entry $D^{d-1,d}_j(P)$.
For a further example consider the following problem, studied in \cite{circles2, circles1, ramos}: let $P$ be a set of $n$ points in general position in the plane.
Are there always two points in $P$ such that any circle through both of them contains at least $\frac{n}{4}$ points of $P$ both inside and outside of the circle?
Using parabolic lifting, proving that for any point set $P$ of size $n$ in convex position in $\mathbb{R}^d$, we have $D^{1,3}_{n/4}(P)>0$ would imply a positive answer to the above question \cite{ramos}.

\section{The condition is necessary}
\label{sec:necessary}

The goal of this section is to show that the condition $\sum_{j=1}^{i-1} D^{0,d}_j \geq 2i + d - 3$ for all $i \geq 2$ with $D^{0,d}_i>0$ is necessary for any depth histogram. To prove this, we first give an upper bound on the depth of any point. We will later give examples of configurations of point sets showing that this result is in fact tight.
Further, we apply similar arguments to show necessary conditions for general depth histograms.

\begin{lemma}
\label{lm:maxdepth}
For any point set $P \subseteq \R^d$ and any point $p \in P$ we have $\emph{td}_P(p) \leq \frac{n-d+2}{2}$.
\end{lemma}

\begin{proof}
Let $P \subseteq \R^d$ and let $p \in P$ be any point with $\emph{td}_P(p) = k$. We will now show that any such point set consists of at least $2k - 2 + d$ points. This then proves the lemma as any point of larger depth would lead to a point set containing more than $n$ points. 

Consider a witnessing halfspace $h_p$ of $P$ and its bounding hyperplane $h$. Note that since $h_p$ is a closed halfspace, $h \subseteq h_p$. Further, we can assume that $h$ contains $p$, as otherwise we can translate $h$ and get a new halfspace $h_p'$ containing at most the same points (possibly even fewer). Additionally, $h$ cannot contain any other point from $P$, that is, otherwise we could rotate $h$ around $p$ slightly, and get a halfspace containing fewer than $k$ points. 

Following the depth of $p$, we know that $h_p$ contains $k$ points. Therefore the question is how many points there are in $\R^d \setminus h_p$. If there are fewer than $k-1$, we directly have a contradiction to the depth of $p$. On the other hand, if there are fewer than $k+d-2$, we also get a contradiction. Note that we can rotate $h$ in any direction until one point, say $q$, changes halfspaces. If $q$ was in $h_p$ before, then we found a halfspace containing $p$ and $k-1$ points in total, contradicting the depth of $p$; thus, $q$ was in $\R^d \setminus h_p$. We can do this until there are $d$ points on $h$ (one of them being $p$) and both halfspaces still need to contain at least $k-1$ points. Thus, in total we need to have $2 (k-1) + d$ points. 
\end{proof}

A similar line of reasoning can be applied to $k$-faces and convex Tukey depth:

\begin{proposition}
Let $P$ be a point set in $\mathbb{R}^d$ and assume that $P$ spans a $k$-face $F$ with $\emph{ctd}_P(F)=i$.
Then $P$ spans at least $2\binom{i-k-1}{m+1}$ $m$-faces of smaller depth.
In other words, for any depth histogram $cD^{k,d}$ and all nonzero entries $cD^{k,d}_i$ with $i \geq 2$ we have \[\sum_{j=1}^{i-1} cD^{m,d}_j \geq 2\binom{i-k-1}{m+1}.\] 
\end{proposition}

\begin{proof}
Consider a witnessing halfspace $h_F$ of $F$ and its bounding hyperplane $h$.
As $\emph{ctd}_P(F)=i$ and $F$ is spanned by $k+1$ points, the halfspace $h_F$ contains $i-k-1$ other points.
Looking at the complement of $h_F$ and translating $h$, we can find another halfspace $h_2$ containing $i-k-1$ points of $P$, with $h_F\cap h_2=\emptyset$.
We have thus found two disjoint subsets of $P$, each of size $i-k-1$.
Further, each $m$-face spanned by $m+1$ points in a subset has depth at most $i-1$, as witnessed by a translation of $h_F$ or $h_2$, respectively.
\end{proof}

The second property that we need in this section is that we can delete points of high depth without changing the depth of points of lower depth.
\begin{lemma}\label{lm:removehighdepth}
For any point set $P \subseteq \R^d$ and any two points $p,q \in P$ with $\emph{td}_P(p) \leq \emph{td}_P(q)$ we have $\emph{td}_{P}(p) = \emph{td}_{P \setminus q}(p)$.
\end{lemma}

\begin{proof}
Let $p,q \in P$ be two points of any point set such that $\emph{td}_P(p) \leq \emph{td}_P(q)$. Consider any witnessing halfspace $h_p$ of $p$. If $q \in h_p$ then we can translate $h_p$ until $q$ lies on the boundary; let us denote the new halfspace as $h_q$. Clearly $h_q$ contains at most $|h_p|-1$ points of $P$ (as the point $p$ is not in there). Therefore $\emph{td}_P(p) > \emph{td}_P(q)$ which is a contradiction. Hence, $q \notin h_p$ and thus deleting $q$ can never change the depth of $p$. 
\end{proof}

This lemma has direct applications for histograms, that is, by repeatedly applying the lemma one can easily show the following.

\begin{proposition}\label{prop:removingpoints}
For any depth histogram $[a_1,a_2,\ldots,a_{m-1},a_m]$, both $[a_1,a_2,\ldots,a_{i-1}, a_i]$ and $[a_1,a_2,\ldots,a_{i-1},1]$ with $i \leq m$ are depth histograms. 
\end{proposition}

Combining both, the knowledge about the maximal possible depth and that we can remove deep points allows us to prove that the condition in Theorem \ref{th:defhisto} is necessary. For ease of reading here the necessary condition of the theorem once more. 
\begin{corollary}[Necessary condition of Theorem \ref{th:defhisto}]
\label{cor:necessary}
For any depth histogram $D^{0,d}$ and all nonzero entries $D^{0,d}_i$ with $i \geq 2$ we have \[\sum_{j=1}^{i-1} D^{0,d}_j \geq 2i + d - 3.\] 
\end{corollary}

\begin{proof}
For the sake of contradiction, let us assume that there exists a depth histogram $D^{0,d}$ for which the statement is not true. This means, there is a nonzero entry $D^{0,d}_i$ and $\sum_{j=1}^{i-1} D^{0,d}_j < 2i + d - 3$.

Using Proposition \ref{prop:removingpoints}, we can cut off the depth histogram $D^{0,d}$ at any point and so we can consider the histogram $D':= [D^{0,d}_1,\ldots, D^{0,d}_{i-1},1]$. We can now bound from above the number of points in the point set $P'$ corresponding to this histogram. From the assumption above we know that there are fewer than $2i + d - 3$ points of depth at most $i$ and by definition there is only one point of depth $i$. Thus, we have $|P'| < 2i + d - 2$. This however contradicts our observation about the maximum possible depth of any point. By Lemma \ref{lm:maxdepth}, we know that any point in $P'$ has depth less than $\frac{(2i + d - 2) - d + 2}{2} = i$. This contradicts the fact that we have a point of depth $i$, that is, $D'_i = 1$; and altogether this proves the corollary.
\end{proof}

From this, we can get a necessary condition for the general case:

\begin{corollary}
For any depth histogram $D^{k,d}$ and all nonzero entries $D^{k,d}_i$ with $i \geq 2$ we have \[\sum_{j=1}^{i-1} D^{0,d}_j \geq 2i + d + k - 3.\] 
In other words, if a point set $P$ in $\mathbb{R}^d$ spans a $k$-face $F$ with $\text{atd}_P(F)=i$, then $P$ contains at least $2i+d+k-3$ points of smaller depth.
\end{corollary}

\begin{proof}
Let $F^{\perp}$ be the orthogonal complement of the affine hull of $F$, and let $\pi$ be the orthogonal projection from $\mathbb{R}^d$ to $F^{\perp}$.
Then $\pi$ maps $P$ to a point set $P'$ in $F^{\perp}$, where the $k+1$ points that span $F$ get mapped to the same point $p_0$.
In $F^{\perp}$, we have $\emph{td}_{P'}(p_0)=i$, so by Corollary \ref{cor:necessary}, $P'$ has at least $2i+d-3$ points of smaller depth.
As $k+1$ points were mapped to $p_0$ and for every point $p$ we have $\emph{td}_{P'}(p)\geq \emph{td}_{P}(p)$, the result follows.
\end{proof}

\subsection{Two special configurations}
\label{sec:symmetric}

Before proving that the condition is sufficient, we make a small detour and revisit Lemma \ref{lm:maxdepth} about the maximum possible depth. It is worth noting that the bound given in the theorem is tight.
We will show this in some detail using point sets in so-called \emph{symmetric configuration} \cite{Ruiz-Vargas2017}, as these point sets will also be at the core of our proof that the condition of Theorem \ref{th:defhisto} is sufficient.

\begin{definition}\label{def:symwheel}
A point set $P \subseteq \R^d$ in general position is in 
\begin{enumerate}
\item \emph{symmetric configuration} if and only if there exists a \emph{central point $c \in P$} such that every hyperplane through $c$ and $d-1$ other points of $P$ separates the remaining points into two halves of equal size.
\item \emph{eccentric configuration} if and only if there exists a \emph{central point $c \in P$} such that every hyperplane through $c$ and $d-1$ other points of $P$ \emph{almost} separates the remaining points into two halves of equal size, that is, divides the remaining points in two sets with difference in cardinality of at most 1.
\end{enumerate}
\end{definition}

Note that depending on the dimension and the size of $P$, only one of the definitions can be applied. We call the point sets in symmetric (or eccentric configuration, respectively), also \emph{symmetric point sets} for short (\emph{eccentric point sets}, respectively). Examples of such point sets are given in Figure \ref{fig:symmetricpointset}. 
\begin{figure}[h]
\centering
	\begin{subfigure}{.49\linewidth}
		\centering
		\includegraphics[scale = 1]{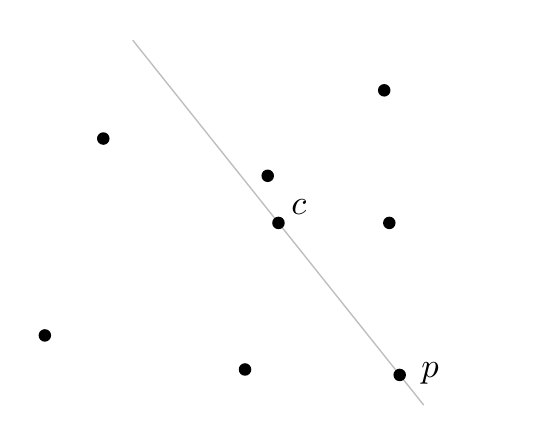}
	\begin{minipage}{.1cm}
	\vfill
	\end{minipage}
        \end{subfigure}
        \hfill
	\begin{subfigure}{.49\linewidth}
		\centering
		\includegraphics[scale = 1]{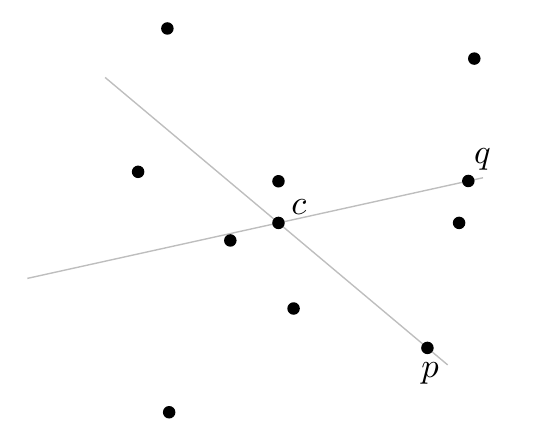}
	\begin{minipage}{.1cm}
	\vfill
	\end{minipage}		
	\end{subfigure}
	\caption{Two point sets in symmetric and eccentric configuration, respectively. The lines through $c$ and $p$ or $c$ and $q$, respectively, (almost) divide the remaining point set.}
	\label{fig:symmetricpointset}
\end{figure}

\begin{lemma}\label{lm:depthsymwheel}
The symmetric central point $c$ in a symmetric (or eccentric) point set $P$ has depth $\emph{td}_P(c) = \lfloor \frac{n-d+2}{2} \rfloor$.
\end{lemma}

\begin{proof}
Note that it is enough to consider hyperplanes having $d$ points on the boundary to completely determine the depth of any point in the point set, as rotating then ensures halfspaces in the desired form. Since we are only interested in the depth of the symmetric central point, we only consider hyperplanes having $c$ on the boundary. 

If we have a point set in symmetric configuration, then by definition, the hyperplanes divide the point set into two sets of the same size. As there are $d$ points on the hyperplane, there are $\frac{n-d}{2}$ strictly on either side. Hence, slightly rotating gives us a halfspace containing $c$ and $\frac{n-d}{2}$ points of $P$ in total and therefore we have $\emph{td}_P(c) = \frac{n-d+2}{2}$ as claimed.

For the case of having an eccentric point set, the hyperplanes under consideration contain $c$ and $d-1$ other points of $P$ and they divide the remaining points into two sets of almost the same size. In particulare we find $\lfloor \frac{n-d}{2} \rfloor$ on one side of the hyperplane and $\lceil \frac{n-d}{2} \rceil$ on the other side. Thus, slightly rotating gives a closed halfspace containing $c$ and $\lfloor \frac{n-d}{2} \rfloor$ other points and therefore we have $\emph{td}_P(c) = \frac{n-d+1}{2}$ as claimed.
\end{proof}

At first glance, it is not clear that symmetric and eccentric point sets of any size exist in any dimension.
We will show that they do in the next section, this will be an important step in proving that the condition of Theorem \ref{th:defhisto} is sufficient.


\section{The condition is sufficient}
\label{sec:sufficient}

To prove that the condition we gave in Theorem \ref{th:defhisto} is sufficient, we build up point sets according to their histograms by adding points one-by-one. In other words, given a histogram, we start from the points in convex position (as many as there are of depth $1$). We then add new points at places, where they have the maximal possible depth, that is, we will add them in the ``center'' of the point set. We then push them outwards until they have the right depth, without changing the depth of any other point. In this way we successively add all points of depth 2, then the ones of depth 3 and so on. 

In this chapter we thus show what happens to the histogram when pushing points outwards (Section \ref{sec:moving}) and where to add new points and in which direction we push them (Section \ref{sec:inserting}). We first show, that we can indeed move deepest points outwards such that they can reach all required depths. In particular, this movement is possible without changing the depth of any point except the moved one. 

\subsection{Moving points}
\label{sec:moving}

First, we make an easy observation that is key to see how moving points affects the Tukey depth histogram of a point set. 

\begin{observation}\label{obs:ordertype}
The depth of a point $q \in P$ can only change if the order type of the point set changes. 
\end{observation}

Note that the Tukey depth of $q$ can only change if $q$ is involved in the change in the order type. In other words, $q$ was pushed over a hyperplane formed by $d$ other points of the point set. We now formally characterize what happens in any such case.

\begin{proposition}\label{lm:facechange}
Let $P \in \R^d$ be a point set and $q \in P$ be an arbitrary point. Let $q'$ be a point close to $q$, such that the order types of $P$ and $P' := P \setminus \{q\} \cup \{q'\}$ only differ in one simplex $\s$, that is, $\s := conv \{p_1, \ldots, p_d, q\}$ and $\s'$, respectively. Let $h$ be the hyperplane spanned by $p_1,\ldots,p_d$ and $\q$ the intersection of $h$ with the line $qq'$.
\begin{itemize}
\item If $\q \notin conv\{p_1,\ldots,p_d\}$, then $\emph{td}_P(q) = \emph{td}_{P'}(q')$, and
\item otherwise, if $\q \in conv\{p_1,\ldots,p_d\}$, then $|\emph{td}_P(q) - \emph{td}_{P'}(q')| \leq 1$. 
\end{itemize}
\end{proposition}

\begin{proof}
For simplicity of notation, let us denote $p_1, \ldots, p_d$ as $\A$.
Let $h_1$ be any hyperplane spanned by $d-1$ points of $P$ and $q$ and $h_2$ the hyperplane spanned by the same $d-1$ points and $q'$ instead of $q$.
Note that by the assumption on the order types, $h_1$ and $h_2$ have the same points of $P$ ($P'$, resp.) above (below, resp.), with the only exception when all the $d-1$ points lie in $\A$.
Thus, the only hyperplanes we need to consider are the ones spanned by $d-1$ points of $\A$ and $q$ or $q'$, respectively.

Let us therefore denote by $h_1^i$ and $h_2^i$ the hyperplanes spanned by $\A \setminus p_i$ and $q$ and $q'$, respectively.
The only point that lies above $h_1^i$ and below $h_2^i$ (or vice versa) is $p_i$, again using the assumption on order types.
Consequently, if there exist hyperplanes of both forms, that is, some with $p_i$ above $h_1^i$ and some with $p_i$ below $h_1^i$, then the depth of $q$ in $P$ is the same as the depth of $q'$ in $P'$, and it changes by at most 1 otherwise.
Of which form the hyperplanes $h_1^i$ and $h_2^i$ are is determined by whether the affine hull of $\A \setminus p_i$ separates $p_i$ and $\q$ in $h$.
In particular, all hyperplanes are of the same form if and only if $\q \in conv\{p_1,\ldots,p_d\}$.
Thus, the depth of $q$ can only change in this case, and if so, then at most by 1.
\end{proof}

Note that not only do we know what happens to the depth of $q$ but whenever $q$ has the highest depth among all points, we also know that the depths of the other points do not change.
\begin{observation}\label{obs:highestpointmoved}
Whenever we have $\emph{td}_P(q) > \emph{td}_P(p)$ for all points $p$ in the point set, then $\emph{td}_P(p) = \emph{td}_{P'}(p)$.
\end{observation}

The observation is a direct corollary of Lemma \ref{lm:removehighdepth} by first removing $q$ and then reinserting $q'$. Thus, we now exactly know how the Tukey depth histogram behaves when moving points of large depths.

\subsection{Inserting a new point}
\label{sec:inserting}

We have already seen point sets, that contain a point of maximum possible depth. These special point sets will help us placing new points of large depth, which we then can push outwards. 

For this, let $P$ be a point set in general position and in symmetric (eccentric, respectively) configuration missing the symmetric central point. If we place a new point $p$ at the location of the (previously inexistent) symmetric central point, then by Lemmas \ref{lm:maxdepth} and \ref{lm:depthsymwheel}, we know that $p$ has the maximal possible depth. Now, we are able to push $p$ outwards until it has the desired depth and the resulting point set is in eccentric (symmetric, respectively) configuration. An example of what happens in dimension two can be found in Figure \ref{fig:pushingp1}. It is pretty easy to see that in dimension two, this always works.
\begin{figure}[h]
\centering
	\begin{subfigure}{.24\linewidth}
		\centering
		\includegraphics[scale = 1]{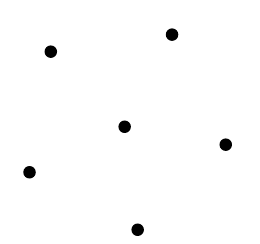}
	\begin{minipage}{.1cm}
	\vfill
	\end{minipage}
        \end{subfigure}
        \hfill
	\begin{subfigure}{.24\linewidth}
		\centering
		\includegraphics[scale = 1]{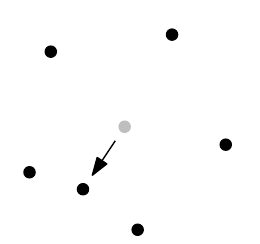}
	\begin{minipage}{.1cm}
	\vfill
	\end{minipage}
        \end{subfigure}
        \hfill
	\begin{subfigure}{.24\linewidth}
		\centering
		\includegraphics[scale = 1]{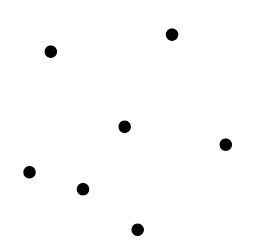}
	\begin{minipage}{.1cm}
	\vfill
	\end{minipage}
	\end{subfigure}
        \hfill
	\begin{subfigure}{.24\linewidth}
		\centering
		\includegraphics[scale = 1]{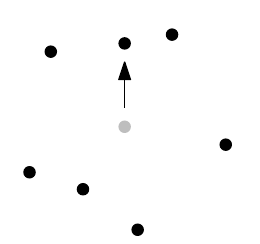}
	\begin{minipage}{.1cm}
	\vfill
	\end{minipage}
	\end{subfigure}	
	\caption{A point set in symmetric configuration (left). After pushing the symmetric central point out (second from left), we arrive at a point set in eccentric configuration missing the symmetric central point. Adding a new point at maximum possible depth (third from left). Pushing out again gets us back into a symmetric point set missing the symmetric central point (rightmost).}
	\label{fig:pushingp1}
\end{figure}

\begin{lemma}\label{lm:symeccdimtwo}
For any point set $P \subseteq \R^2$ in general position and in eccentric (symmetric, resp.) configuration there exists a direction in which we can push the central point such that after adding a new center we have a symmetric (eccentric, resp.) point set in general position. 
\end{lemma}

\begin{proof}
First, note that if $P$ is in symmetric configuration, any direction does the job. The only ambiguity to be careful about is choosing a direction such that the resulting point set (after adding a new center) is again in general position. However, this is always possible. 

If $P$ is eccentric, then there exist two neighbors in the rotational order of points around $q$ without a symmetric central line dividing them. Let us denote these points as $p_1$ and $p_2$ and choose to move $q$ outwards on an ''opposite'' halfline, see Figure \ref{fig:pushingp}, right. This ensures that the point set is again nicely symmetric around the new center (i.e.~the line rotational order is alternating between points passing in front of, and behind $q$). In particular, every symmetric central line is halving the point set, thus the point set is symmetric. 
\end{proof}

\begin{figure}[h]
\centering
	\begin{subfigure}{.49\linewidth}
		\centering
		\includegraphics[scale = 1]{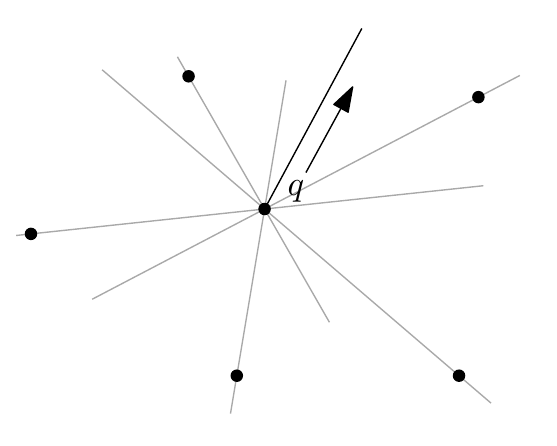}
	\begin{minipage}{.1cm}
	\vfill
	\end{minipage}
        \end{subfigure}
        \hfill
	\begin{subfigure}{.49\linewidth}
		\centering
		\includegraphics[scale = 1]{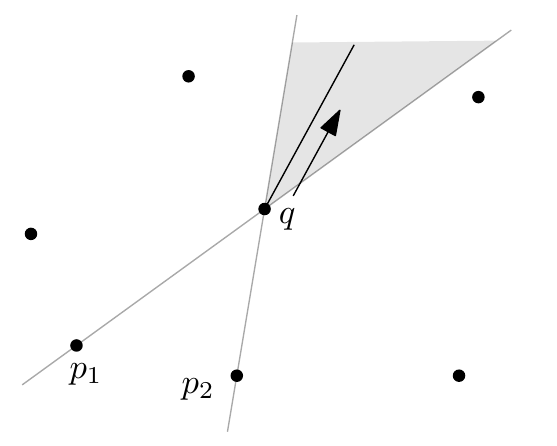}
	\begin{minipage}{.1cm}
	\vfill
	\end{minipage}		
	\end{subfigure}
	\caption{The central point and the direction in which we push it if the point set is symmetric (left) and if the point set is eccentric (right).}
	\label{fig:pushingp}
\end{figure}

In higher dimensions it is not so easy to see how to get the directions and why they always exist. We will do an induction argument on the dimension of the point set; to understand the necessary ideas we start in three-dimensional space. Let us further denote a point set as \emph{spherical}, if every point (except maybe one central point) lies on a sphere around the origin. We further extend the definition of symmetric (eccentric, resp.) point sets to spherical point sets, that is, they are symmetric (eccentric, resp.) with respect to the origin (instead of a symmetric central point).

\begin{proposition}\label{prop:spherical3dim}
For every spherical symmetric point set $P \subseteq \R^3$ in general position there exist two points $p_1, p_2 $ such that adding any one of them to $P$ results in a spherical eccentric point set and adding both, $p_1$ and $p_2$, results in a spherical symmetric point set in general position.
\end{proposition}

Here general position means that no three points of $P$ are on a common plane through the origin.
In particular, no two points are exactly opposite on the sphere. The reason we add two points at a time is that they heavily depend on one another.

\begin{proof}
Let us assume that $P$ is a spherical symmetric point set. Note that adding any point on the sphere to $P$ results in an eccentric point set by definition (that is clearly spherical). The harder part of the proof is to show that there exist two points such that adding both results in a symmetric point set again. For this, let us assume without loss of generality that there is neither a point at the south pole nor at the north pole of the sphere. Let $p_S$ be at the south pole of the sphere, $p_N$ at the north pole and define $P^S := P \cup p_S$, $P^N := P \cup p_N$ and $P^{SN} := P \cup p_S \cup p_N$. Note that $P^{SN}$ is symmetric; however, it is not in general position. Therefore we will slightly move $p_N$, such that we get a point set that is both, symmetric and in general position. In the following, we will construct this new point, denoted as $p_{N'}$.

For this let $Q \subseteq \R^2$ be the stereographic projection of $P$ from the north pole, meaning $Q$ is a point set in the plane where the projection of the north pole coincides with the origin of the plane. For all $p_i \in P$ we denote their projections as $q_i$. Further, let $q_N$ be at the origin and define $Q^N := Q \cup q_N$; in other words $q_N$ and $Q^N$ are the projections of $p_N$ and $P^N$, respectively.

We can now show that $Q^N$ is eccentric (with respect to $q_N$). Every line through $q_N$ and $q_i \in Q^N$ corresponds to a plane in $\R^3$ through $p_N$, $p_i$ and $p_S$. Therefore the plane goes through the origin and two points of $P^{N} \subseteq \R^3$ and is almost halving the point set $P^N$ (recall that $P^N$ is eccentric). Hence, the line is also almost halving the point set $Q^N$ (the points above the plane are on the same side of the line) and so the point set is indeed eccentric. 

We can now push $q_N \in Q^N$ slightly into some direction using Lemma \ref{lm:symeccdimtwo}. Let us denote the new point as $q_{N'}$ and its projection back into $\R^3$ as $p_{N'}$. Note that we not only get a direction from Lemma \ref{lm:symeccdimtwo}, but it ensures also that $Q^{NN'} := Q^N \cup q_{N'}$ is symmetric (with respect to $q_N$). 

Define $P^{SN'} := P^S \cup p_{N'}$ and note that it is in general position. It remains to show that $P^{SN'}$ is symmetric. For this let $h$ be any plane through the origin and two arbitrary points $p_1, p_2 \in P^{SN'}$. By definition, we need to show that $h$ is halving the point set $P^{SN'}$. 

If $p_1,p_2 \in P$ then it is easy to see that $h$ halves the point set $P$ and therefore also $P^{SN'}$ (we added one point above $h$, namely $p_{N'}$, and one point below $h$, that is $p_S$). 

If $p_1 = p_S$ and $p_2 \in P$ then the plane $h$ also goes through $p_N$. Thus $h$ corresponds to a line $l$ in $\R^2$ that goes through the origin ($q_N$) and $q_2 \in Q^{NN'}$. Remember that $Q^{NN'}$ is symmetric (with respect to $q_N$), therefore $l$ is halving $Q^{NN'}$ and consequently $h$ is halving the point set $P^{SN'}$. 

The last possibility is that $p_1 = p_{N'}$ and $p_2 \in P$ (and in particular $p_2 \neq p_S$). Here we give an indirect proof, so let us assume that $h$ does not half the point set $P^{SN'}$, i.e.~there are at least $k+2$ points on one side and at most $k$ on the other.
Consider the line $\ell \in \R^3$ through the origin and $p_2$, and let $\ell^{\perp}$ be a plane orthogonal to it.
Under the orthogonal projection $\pi$ to $\ell^{\perp}$, $P^{SN'}$ gets mapped to a point set of the same size.
In particular, each plane through the origin and $p_2$ gets mapped to a line through $\pi(p_2)$, and two points $a$ and $b$ are on different sides of the plane if and only if $\pi(a)$ and $\pi(b)$ are on different sides of the projected line.
By our assumption we thus know that the line through $\pi(p_2)$ and $\pi(p_{N'})$ splits the point set into unequal parts.
It is a well known property of planar point sets that then there must be some point $\pi(p_3)$ such that the line through $\pi(p_2)$ and $\pi(p_3)$ splits the point set in the same ratio.
In fact, for every unequal ratio, the number of lines through $\pi(p_2)$ that split the line in this ratio must be even, see e.g.~Lemma 3.5 (ii) in \cite{Ruiz-Vargas2017}.
But then the plane (in $\R^3$) through the origin, $p_2$ and $p_3$ does not half the point set, which is a contradiction.
\end{proof}

Note that this proof heavily relied on the facts that we have a spherical point set and that we can find the needed directions for point sets in $\R^2$. While the former condition can easily be avoided (see Lemma \ref{lm:general3dim} below), the latter can be ensured with doing an induction over the dimension. 

Formally, for a point set $P$ that is not spherical but in symmetric or eccentric configuration, let $S$ be a surrounding sphere of $P$ with center $q$ (the symmetric central point of $P$). For all points $p \in P$ such that $p \neq q$, we push $p$ out onto $S$ on the line $pq$ and denote the resulting point set as the \emph{induced spherical point set $P'$}. Note that $P'$ is clearly spherical and we additionally know the following.
\begin{observation}\label{obs:spherical}
The induced spherical point set $P'$ of $P$ is symmetric (eccentric, respectively) if and only if $P$ is symmetric (eccentric, respectively). 
\end{observation}

This follows from the construction of $P'$ as all necessary hyperplanes remain unchanged and in particular they split the point sets $P$ and $P'$ in the exact same way. 

\begin{lemma}\label{lm:general3dim}
For every symmetric point set $P \subseteq \R^3$, there exist two directions $v_1$ and $v_2$ such that we can push the central point into either direction; add a new central point and arrive at an eccentric point set $P'$. We can then push the newly added point into the other direction, and arrive at a symmetric point set $P''$ missing the symmetric central point. 
\end{lemma}

\begin{proof}
The idea is to use Proposition \ref{prop:spherical3dim} on the induced spherical point set of $P$ to get directions in which we could push the points. Let $SP$ be the induced spherical point set of $P$. By Proposition \ref{prop:spherical3dim} there exist two positions $s_1$ and $s_2$ where we can add new points to $SP$ such that the resulting point sets are eccentric (symmetric, resp.). Define $v_1$ to be the directed line from the origin to $s_1$ and similarly, define $v_2$ as the directed line to $s_2$. 

Note that pushing the symmetric central point in $P$ in direction $v_1$ results in the same induced spherical point set, that is $SP \cup s_1$, independently of how far we push. Since the induced spherical point sets are symmetric (eccentric, respectively) the same is true for the point sets $P'$ and $P''$, using Observation \ref{obs:spherical}.
\end{proof}

Note that the exact same arguments work not only in $\R^3$ but in arbitrary dimensions. The plan is therefore to prove statements similar to Proposition \ref{prop:spherical3dim} and Lemma \ref{lm:general3dim} in $d$ dimensions. We will give an inductive proof over the dimension.

\begin{theorem}\label{thm:directions}
For every symmetric point set $P \subseteq \R^d$, there exist two directions $v_1$ and $v_2$ such that we can push the central point into either direction; add a new central point and arrive at an eccentric point set $P'$. We can then push the newly added point into the other direction, and arrive at a symmetric point set $P''$ missing the symmetric central point.
\end{theorem}

\begin{proof}
As already mentioned, we do a proof by induction over the dimension $d$. For $d=2$, Lemma \ref{lm:symeccdimtwo} gives a proof (of an even stronger statement), therefore let us assume that $d \geq 3$ and the theorem is proven for dimension $d-1$ already. 

First, we can show that for every spherical symmetric point set $P \subseteq \R^d$ in general position there exist two points $p_1, p_2 $ such that adding any one of them to $P$ results in a spherical eccentric point set and adding both, $p_1$ and $p_2$, results in a spherical symmetric point set in general position. The points $p_1$ and $p_2$ will again be the south pole and very close to the north pole, respectively. As the arguments to prove this statement are exactly the same as the ones explained in detail in the proof of Proposition \ref{prop:spherical3dim}, we only sketch them here.

Adding any one of the points to $P$ results in an eccentric point set by definition. Similarly to the situation in $\R^3$ we can add the south pole and the north pole to $P$ and do a stereographic projection from the $d$-dimensional sphere to $(d-1)$-dimensional Euclidean space, that is, for a spherical eccentric point set in $\R^d$ we get an eccentric point set in $\R^{d-1}$. By the induction hypothesis there exists a direction in which we can push the central point (that corresponds to the north pole) and we get a point close to the north pole that can be added to $P$. One can again show that the resulting point set is symmetric. Note that in the very end, instead of a line $\ell$ we get a $(d-2)$-flat, but we still get a projection on a ($2$-dimensional) plane and so the reasoning is the same.

Lastly, following the exact reasoning of the proof of Lemma \ref{lm:general3dim} we can extend the result to all symmetric sets in $\R^d$.
\end{proof}

\subsection{Putting everything together}

Recall that we want to prove that the condition given in Theorem \ref{th:defhisto} is sufficient.

\begin{theorem}[Sufficient condition of Theorem \ref{th:defhisto}]
A vector $D^{0,d}$ satisfying $\sum_{j=1}^{i-1} D^{0,d}_j \geq 2i + d - 3$ for all nonzero entries $D^{0,d}_i$ with $i \geq 2$ is a depth histogram of a point set in general position in $\R^d$.
\end{theorem}

\begin{proof}
If all entries of $D^{0,d}_i$ with $i \geq 2$ are zero, then let $P$ be a point set of $D^{0,d}_1$ many points in general, convex position in $\R^d$. This proves that the vector is a histogram. Let us therefore assume that there is at least one nonzero entry in $D^{0,d}_i$ with $i \geq 2$. Let $P$ be the vertices of a simplex in $\R^d$ around the origin and note that $P$ is a (spherical) symmetric point set. We will now add points to $P$ (in pairs) and maintain the condition that $P$ is a symmetric point set. We first add all points of depth one, then all points of depth two and so on. 

Assume that $P$ consists of $n$ points and assume further that there are points missing in $P$ (i.e.~$P$ does not have histogram $D^{0,d}_i$). Let us denote the smallest missing depth by $j$ and note that this means that all points in $P$ have depth at most $j$. We now add a point $p$ in the origin to $P$. Note that $p$ has depth $\lfloor \frac{(n+1)-d+2}{2} \rfloor$, see Lemma \ref{lm:depthsymwheel}. By the condition of the Theorem, we know that $n \geq 2j + d - 3$ and thus $j \leq \frac{n-d+3}{2}$. Therefore we can push $p$ outwards into a direction given by Theorem \ref{thm:directions}. Note that this is why we add ``pairs of points''. While pushing $p$ outwards, at some point the order type of the point set $P$ changes (recall that we added $p$ to $P$) and the depth of $p$ may change. We continue pushing $p$ until it has depth $j$. Proposition \ref{lm:facechange} and Observation \ref{obs:highestpointmoved} guarantee that the only depth that changed while moving $p$ is the one of point $p$, as all other points of the point set have lower depth. Theorem \ref{thm:directions} gives us not only the needed direction but also shows that we can maintain the property of having symmetric (and eccentric) point sets throughout the whole process. 
\end{proof}

\section{Number of depth histograms}
\label{sec:number}

The characterization of Tukey depth histograms $D^{0,d}(P)$ allows to compute the exact number of different histograms for point sets consisting of $n$ points in $\R^d$. 

\begin{definition}\label{def:number}
Let $D(n,d)$ denote the number of different Tukey depth histograms $D^{0,d}(P)$, for point sets $P \subseteq \R^d$ consisting of $n$ points. Further, let $D(n,d,l)$ denote the number of different histograms of points, for point sets in $\R^d$ consisting of $n$ points with the deepest point having depth $l$.
\end{definition}

Clearly we can sum up over all possible values of the deepest point to get the total number of different Tukey depth histograms, that is, using Lemma \ref{lm:maxdepth} we have
\begin{equation}
\label{eq:numbers}
D(n,d) = \sum_{j=1}^{\lfloor \frac{n-d+2}{2} \rfloor}{D(n,d,j)}.
\end{equation}
Note that for some values of $d$, $n$ and $l$ it is easy to see how to compute $D(n,d,l)$. Point sets in the plane consisting of three points can only have one histogram; that is, $D(3,2) = 1$. Similarly, when the deepest point has depth $1$ we also have exactly one possible histogram; i.e.~$D(n,2,1) = 1$. For arbitrary values of $d$, $n$ and $l$ we can prove the following formula.

\begin{theorem}\label{thm:number}
For any dimension $d \geq 2$, any $n \geq d+1$ and any suitable $l$, that is, $l \leq \frac{n-d+2}{2}$, we have
\begin{equation}
\label{eq:thmnumber}
D(n,d,l) = \frac{(n - 2l - d + 3)(n+l-d-1)!}{(l-1)!(n-d+1)!}
\end{equation}
\end{theorem}

We will prove the Theorem inductively over $l$; however, to complete the proof, let us first give some of the characterizations and calculations that will be needed. 

\begin{lemma}\label{lm:def}
For any $d \geq 2$, any $l$ and $n \geq 2l + d - 2$ we have,
\begin{equation}
\label{eq:lmdeftext}
D(n,d,l) = \sum_{i=1}^{l}{D(n-1,d,i)} = \sum_{j=2l + d - 2}^{n}{D(j,d,l-1)}
\end{equation}
\end{lemma}

\begin{proof}
The proof of this Lemma is essentially applying definitions and rearranging sums and does not need any involved ideas. First of all, observe that by removing a ``deepest'' point from a point set consisting of $n$ points we always get a point set consisting of $n-1$ points. Certainly the maximal depth does not increase when removing a point, and by Lemma \ref{lm:removehighdepth} no other depth changes; thus for any $d \geq 2$, any $l \geq 1$ and any $n \geq 2l + d - 2$ we have the following
\begin{equation}\label{eq:basicsum}
D(n,d,l) = \sum_{i=1}^{l}{D(n-1,d,i)}.
\end{equation}

Additionally note that whenever $n < 2l + d - 2$, we cannot have a point of depth $l$, see Lemma \ref{lm:maxdepth}. Therefore we can compute in the following way
\begin{align*}
D(n,d,l) &\stackrel{\eqref{eq:basicsum}}{=} \sum_{i=1}^{l}{D(n-1,d,i)} = \sum_{i=1}^{l-1}{D(n-1,d,i)} + D(n-1,d,l) \\
&\stackrel{\eqref{eq:basicsum}}{=} \sum_{i=1}^{l-1}{D(n-1,d,i)} + \sum_{i=1}^{l}{D(n-2,d,i)} \\
&=\sum_{i=1}^{l-1}{D(n-1,d,i)} + \sum_{i=1}^{l-1}{D(n-2,d,i)} + D(n-2,d,l) \\
&= \, \ldots \\
&= \sum_{i=1}^{l-1}{D(n-1,d,i)} + \ldots + \sum_{i=1}^{l-1}{D(2l+d-3,d,l)} + D(2l + d - 3,d,l) \\
&= \sum_{j=2l + d - 3}^{n-1}{\sum_{i=1}^{l-1}{D(j,d,i)}} + 0 \stackrel{\eqref{eq:basicsum}}{=} \sum_{j=2l + d - 3}^{n-1}{D(j+1,d,l-1)} \\
&=  \sum_{j=2l + d - 2}^{n}{D(j,d,l-1)},
\end{align*}
which concludes the proof of this Lemma.
\end{proof}

\begin{lemma}\label{lm:simplesum}
For any $d \geq 2$, any $k$ and $n \geq 2k + d - 2$ we have,
\begin{equation}
\label{eq:simplesum}
\sum_{j = 0}^{n - 2k - d + 2}{\frac{(j+3)(j+3k - 4)!}{(k-2)!(j+2k-1)!}} = \frac{(n-2k-d +3)(n+k-d-1)!}{(k-1)!(n-d+1)!}
\end{equation}
\end{lemma}

\begin{proof}
We will prove this second Lemma by induction over $n$. For any $d \geq 2$ and any $k$, the base case is given by choosing $n = 2k + d - 2$. In this case the sum shrinks to one term and we have the following:
\begin{equation*}
\begin{split}
\sum_{j = 0}^{n - 2k - d + 2}{\frac{(j+3)(j+3k - 4)!}{(k-2)!(j+2k-1)!}} &= \sum_{j = 0}^{0}{\frac{(j+3)(j+3k - 4)!}{(k-2)!(j+2k-1)!}} = \frac{3 (3k - 4)!}{(k-2)! (2k-1)!}.
\end{split}
\end{equation*}
On the other hand, on the right side of equation \eqref{eq:simplesum}, we get the following:
\begin{align*}
\hspace{-5mm} \frac{(n-2k-d +3)(n+k-d-1)!}{(k-1)!(n-d+1)!} &= \frac{\big((2k + d - 2)-2k-d +3\big)\big((2k + d - 2)+k-d-1\big)!}{(k-1)!\big((2k+d-2)-d+1\big)!} \\
&= \frac{(3k-3)!}{(k-1)!(2k-1)!} = \frac{3 (k - 1) (3k-4)!}{(k-1)(k-2)!(2k-1)!} \\
&= \frac{3(3k-4)!}{(k-2)!(2k-1)!},
\end{align*}
which concludes the base case of the induction.

For the induction step we thus assume that for some $n \geq 2k + d - 2$ equality \eqref{eq:simplesum} is true and we want to prove the statement for the case $n + 1$. The first step is splitting the sum, that is,
\begin{align*}
&\sum_{j=0}^{(n+1) - 2k - d + 2}{\frac{(j+3)(j+3k - 4)!}{(k-2)!(j+2k-1)!}} = \sum_{j=0}^{n - 2k - d + 2}{\frac{(j+3)(j+3k - 4)!}{(k-2)!(j+2k-1)!}} + \, \ldots \\
&\hspace{33mm}\ldots \, + \frac{\big((n+1-2k-d+2)+3\big)\big((n+1-2k-d+2)+3k-4\big)!}{(k-2)!\big((n+1-2k-d+2)+2k-1\big)!}.
\end{align*}

Using the induction hypothesis, we now get that this is equal to 
\begin{align*}
\ldots \, &=  \frac{(n-2k-d +3)(n+k-d-1)!}{(k-1)!(n-d+1)!} + \frac{(n-2k-d+6)(n+k-d-1)!}{(k-2)!(n-d+2)!} \\
&= \frac{(n+k-d-1)!}{(k-1)!(n-d+2)!} \bigg( (n-2k-d+3)(n-d+2) + (n-2k-d+6)(k-1)\bigg) \\
&= \frac{(n+k-d-1)!}{(k-1)!(n-d+2)!} \Big( n^2-2nd +4n - nk +dk + 4k - 4d +d^2-2k^2\Big) \\
&= \frac{(n+k-d-1)!}{(k-1)!(n-d+2)!} \Big( (n-2k-d+4)(n+k-d) \Big) \\
&= \frac{\big((n+1)-2k-d+3\big)\big((n+1)+k-d-1\big)!}{(k-1)!\big((n+1)-d+1\big)!}
\end{align*}

This calculation proves the induction step and with this the Lemma is proven. 
\end{proof}

With these equations we are now able to prove Theorem \ref{thm:number}.

\begin{proof}[Proof of Theorem \ref{thm:number}]
We prove the theorem by induction over $l$. The base case is given by $l=1$ and we already observed that $D(n,d,1) = 1$ for $d \geq 2$ and $n \geq d+1$. The right hand side of equation \eqref{eq:thmnumber} simplifies in the following way
\begin{align*}
\frac{(n - 2l - d + 3)(n+l-d-1)!}{(l-1)!(n-d+1)!} = \frac{(n - d + 1)(n-d)!}{(0)!(n-d+1)!} = 1.
\end{align*}

Hence we have a base case for the induction, therefore let us assume that equation \eqref{eq:thmnumber} holds for suitable $d$ and $n$ and for $l \leq k-1$. We now show that it consequently also holds for $l=k$. We have the following
\begin{align*}
D(n,d,k) &\stackrel{\eqref{eq:lmdeftext}}{=} \sum_{j=2k + d - 2}^{n}{D(j,d,k-1)} \\ 
&\stackrel{I.H.}{=} \sum_{j=2k + d - 2}^{n}{\frac{\big(j - 2(k-1) - d + 3\big)\big(j+(k-1)-d-1\big)!}{\big((k-1)-1\big)!(j-d+1)!} } \\
&= \sum_{j=2k + d - 2}^{n}{\frac{(j - 2k - d + 5)(j+k-d-2)!}{(k-2)!(j-d+1)!} } \\
&= \sum_{j=0}^{n-2k-d+2}{\frac{\big((j+2k+d-2) - 2k - d + 5\big)\big((j+2k+d-2)+k-d-2\big)!}{(k-2)!\big((j+2k+d-2)-d+1\big)!} } \\
&= \sum_{j=0}^{n-2k-d+2}{\frac{(j+3)(j+3k-4)!}{(k-2)!(j+2k-1)!}} \\
&\stackrel{\eqref{eq:simplesum}}{=} \frac{(n-2k-d +3)(n+k-d-1)!}{(k-1)!(n-d+1)!}.
\end{align*}
Note that the assumptions of the theorem on $d$, $n$ and $l$ imply the conditions of the lemmas and so we are indeed always able to use inequalities \eqref{eq:lmdeftext} and \eqref{eq:simplesum}. This concludes the induction step and thus proves the theorem.
\end{proof}

With this we are now also able to compute the number of different histograms of point sets consisting of $n$ points. 
\begin{theorem}\label{thm:numbern}
For any dimension $d \geq 2$ and any $n \geq d+1$, we have
\begin{equation}
\label{eq:thmnumbern}
D(n,d) = \begin{cases}
\frac{2}{n-d+2} {{ 3 \frac{n-d}{2} + 1} \choose {\frac{n-d}{2}}}, &\text{ if $n-d$ is even and} \\
\frac{3}{n-d+2} {{ 3 \frac{n-d}{2} + \frac{1}{2} } \choose {\frac{n-d-1}{2}}}, &\text{ if $n-d$ is odd}
\end{cases}
\end{equation}
\end{theorem}

\begin{proof}
Combining everything we have seen so far, we have
\begin{align*}
D(n,d) &\stackrel{\eqref{eq:numbers}}{=} \sum_{j=1}^{\lfloor \frac{n-d+2}{2} \rfloor}{D(n,d,j)} \stackrel{\eqref{eq:lmdeftext}}{=} D(n+1,d,\lfloor \frac{n-d+2}{2} \rfloor) \\
&\stackrel{\eqref{eq:thmnumber}}{=} \frac{\Big((n+1) - 2 \lfloor \frac{n-d+2}{2} \rfloor - d + 3\Big)\Big((n+1)+\lfloor \frac{n-d+2}{2} \rfloor-d-1\Big)!}{\Big(\lfloor \frac{n-d+2}{2} \rfloor-1\Big)!\Big((n+1)-d+1\Big)!}
\end{align*}

For rounding issues we now split the computation, let us first assume that $n-d$ is even. This simplifies the equation above to the following
\begin{align*}
\ldots \, = \frac{2 \big( 3 \frac{n-d}{2} + 1 \big)!}{\big( \frac{n-d}{2} \big)!\big( n - d + 2\big)!} = \frac{2}{n-d+2} {{ 3 \frac{n-d}{2} + 1 } \choose {\frac{n-d}{2}}},
\end{align*}
whereas $n-d$ being odd results in
\begin{align*}
\ldots \, = \frac{3 \big( 3 \frac{n-d}{2} + \frac{1}{2} \big)!}{\big( \frac{n-d-1}{2} \big)!\big( n - d + 2\big)!} = \frac{3}{n-d+2} {{ 3 \frac{n-d}{2} + \frac{1}{2} } \choose {\frac{n-d-1}{2}}}.
\end{align*}
This completes the proof of Theorem \ref{thm:numbern}.
\end{proof}

\section{Conclusion}
\label{sec:conclusion}

We have introduced Tukey depth histograms of $j$-flats, and discussed how they relate to several problems in discrete geometry.
For general histograms, we have shown necessary conditions for a vector to be a histogram.
For histograms of points, we were able to give a full characterization.
This characterization allowed us to give an exact number of possible histograms.
This is a contrast to other representations of point sets, such as order types, where the exact numbers are not known.

It is an interesting open problem to find better necessary and also sufficient conditions, perhaps even characterizations, of histograms of $j$-flats for $j>0$.
We hope that the ideas and arguments in this paper might be useful in this endeavor.

Another interesting open problem is to relate depth histograms to other representations of point sets.
For example in the plane, the order type determines the $\ell$-vectors for each point, but not vice-versa, that is, there are point sets that have the same sets of $\ell$-vectors but different order types.
Similarly, the set of $\ell$-vectors determines the histograms $D^{0,2}$ and $D^{1,2}$.
Is it true that the reverse is also true or are there point sets for which both $D^{0,2}$ and $D^{1,2}$ are the same but whose sets of $\ell$-vectors are different?

Due to there relation to many problems in discrete geometry, we are convinced that the study of depth histograms has the potential to lead to new insights for many problems.

\section*{Acknowledgments}
Patrick Schnider has received funding from the European Research Council under the European Unions Seventh Framework Programme ERC Grant agreement ERC StG 716424 - CASe.

\bibliographystyle{plainurl}
\bibliography{refs}

\end{document}